\documentclass[twocolumn, journal]{IEEEtran}
\IEEEoverridecommandlockouts

\usepackage{color}
\usepackage{bm}
\usepackage{graphicx}
\usepackage{amsmath}
\usepackage{amssymb}
\usepackage{algorithm}
\usepackage{algorithmic}
\usepackage{multirow}
\usepackage{booktabs}
\usepackage{array}
\usepackage{amsthm}
\usepackage{lipsum}
\usepackage{enumerate}
\usepackage{stfloats}
\usepackage{subfigure}
\usepackage{cases}
\usepackage{diagbox}
\usepackage{xcolor}
\usepackage{xpatch}
\usepackage{cite}

\newtheorem{remark}{Remark}

\newtheorem{defin}{Definition}

\newtheorem{corollary}{Corollary}

\newcommand{\non}{\nonumber}

\title{A Dynamic Grouping Strategy for Beyond Diagonal Reconfigurable Intelligent Surfaces with Hybrid Transmitting and Reflecting Mode}
\author{Hongyu Li, \IEEEmembership{Student Member, IEEE}, Shanpu Shen, \IEEEmembership{Senior Member, IEEE}, and Bruno Clerckx, \IEEEmembership{Fellow, IEEE}
\thanks{Manuscript received October 5, 2022; revised February 26, 2023, and May 7, 2023; accepted  June 18, 2023. This work was supported by Hong Kong Research Grants Council through the Collaborative Research Fund under Grant C6012-20G. The associate editor coordinating the review of this article was Dr. Hongliang Zhang. \textit{(Corresponding author: Shanpu Shen).}}
\thanks{H. Li is with the Department of Electrical and Electronic Engineering, Imperial College London, London SW7 2AZ, U.K. (e-mail: c.li21@imperial.ac.uk).}\\
\thanks{S. Shen is with the Department of Electronic and Computer Engineering, The Hong Kong University of Science and Technology, Clear Water Bay, Kowloon, Hong Kong (e-mail: sshenaa@connect.ust.hk).}
\thanks{B. Clerckx is with the Department of Electrical and Electronic Engineering, Imperial College London, London SW7 2AZ, U.K. and with Silicon Austria Labs (SAL), Graz A-8010, Austria (e-mail: b.clerckx@imperial.ac.uk; bruno.clerckx@silicon-austria.com).}}

\begin{document}

\maketitle
\thispagestyle{empty}
\begin{abstract}
    Beyond diagonal reconfigurable intelligent surface (BD-RIS) is a novel branch of RIS which breaks through the limitation of conventional RIS with diagonal scattering matrices. 
    However, the existing research focuses on BD-RIS with fixed architectures regardless of channel state information (CSI), which limit the achievable performance of BD-RIS.
    To solve this issue, in this paper, we propose a novel dynamically group-connected BD-RIS based on a dynamic grouping strategy. Specifically, RIS antennas are dynamically divided into several subsets adapting to the CSI, yielding a permuted block-diagonal scattering matrix. 
    To verify the effectiveness of the proposed dynamically group-connected BD-RIS, we propose an efficient algorithm to optimize the BD-RIS with dynamic grouping for a BD-RIS-assisted multi-user multiple-input single-output system. 
    Simulation results show that the proposed dynamically group-connected architecture outperforms fixed group-connected architectures. 
\end{abstract}


\vspace{-0.2 cm}
\section{Introduction}
\vspace{-0.2 cm}

Reconfigurable intelligent surface (RIS) has emerged as a revolutionary technique enabling spectrum, cost, and energy efficient communications \cite{wu2019towards,elmossallamy2020reconfigurable}. 
However, most existing works focus on a simple RIS architecture, where each RIS element is connected to a load disconnected from the other elements, yielding mathematically a diagonal phase shift matrix. This architecture enables the phase shift control of the incident waves and can only support the signal reflection, which limits the beam control accuracy and coverage of RIS.

\begin{figure}
    \centering
    \includegraphics[height=1.75 in]{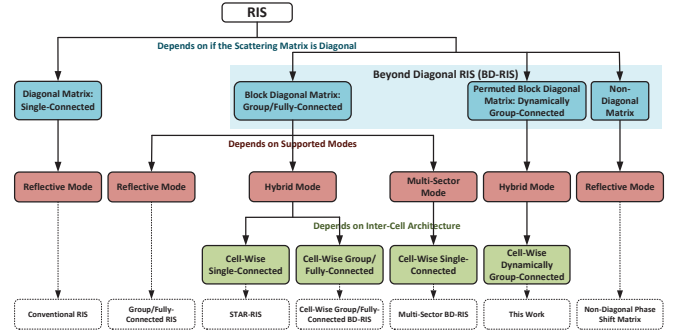}
    \caption{RIS classification tree.}\label{fig:RIS_tree}\vspace{-0.5 cm}
\end{figure}

To break through the limitation of conventional RIS, a novel branch namely beyond diagonal RIS (BD-RIS) as illustrated in Fig. \ref{fig:RIS_tree}, has been recently proposed and investigated. 
The BD-RIS relies on different circuit topologies of antenna ports and thus has scattering matrices beyond diagonal matrices.
Specifically, modeling and architecture design of RIS and microwave theory was first bridged in \cite{shen2021modeling}. 
More general group- and fully-connected architectures were proposed in \cite{shen2021modeling}, leading to block-diagonal scattering matrices and enabling more flexible dual phase shift and amplitude control of the impinging waves.  
Furthermore, optimal design for group/fully-connected BD-RIS was proposed in \cite{nerini2022optimal} and group/fully-connected BD-RIS with discrete-value impedance networks was investigated in \cite{nerini2021reconfigurable}. 
Another novel architecture with a non-diagonal phase shift matrix was proposed in \cite{li2022reconfigurable}. 
However, the BD-RIS models proposed in \cite{shen2021modeling,nerini2022optimal,nerini2021reconfigurable,li2022reconfigurable} are restricted to the reflective mode, that is signals can only be reflected to one side of the surface, resulting a ``waste'' of space.   
To realize a full-space coverage, simultaneously transmitting and reflecting RIS (STAR-RIS) \cite{xu2021simultaneously} or intelligent omni-surface (IOS) \cite{zhang2020beyond} supporting the hybrid mode has been proposed recently as an important extension of RIS. 
The STAR-RIS/IOS is then first practically implemented in \cite{zeng2022intelligent}.
To go deep into the essence of STAR-RIS, a BD-RIS model unifying modes (reflective/transmissive/hybrid) and architectures (single/group/fully-connected) has been established in \cite{li2022generalized}, showing that STAR-RIS is essentially a particular instance of group-connected architecture when group size is equal to 2. 
The work in \cite{li2022generalized} was further extended to a multi-sector mode, which was a generalization of hybrid mode achieving not only full-space coverage, but also significant performance enhancement \cite{li2022beyond}.

The limitation of the proposed BD-RIS architectures in \cite{shen2021modeling,li2022generalized,li2022beyond} is that they are restricted to fixed grouping strategies, where all RIS antennas are uniformly partitioned and adjacently grouped regardless of channel state information (CSI). 
However, exploring different grouping strategies of BD-RIS still remains an open problem.  
To solve this issue, in this paper, we, for the first time, propose a dynamically group-connected BD-RIS model, 
which is a novel branch of BD-RIS as shown in Fig. \ref{fig:RIS_tree}.
Specifically, the proposed model is based on a dynamic grouping strategy, whose main idea is to dynamically partition RIS antennas into multiple groups with variable group sizes adapting to CSI. 
The contributions of this paper are summarized as follows.

\textit{First}, we derive the mathematical model and discuss the practical implementation of the proposed dynamically group-connected BD-RIS.
\textit{Second}, we apply the dynamically group-connected BD-RIS in a multi-user multiple-input multiple-output (MU-MISO) system and propose a novel and efficient algorithm to optimize the BD-RIS with the dynamic grouping strategy.
\textit{Third}, we provide simulation results to evaluate the performance enhancement of the proposed dynamically group-connected BD-RIS compared to fixed group-connected cases.

\textit{Notations}:
Boldface lower- and upper-case letters indicate column vectors and matrices, respectively.
$(\cdot)^T$ and $(\cdot)^H$ denote transpose and conjugate-transpose operations, respectively.
$\Re \{ \cdot \}$ denotes the real part of a complex number.
$\mathbf{I}_L$ denotes an $L \times L$ identity matrix.
$\| \mathbf{A} \|_F$ denotes the Frobenius norm of $\mathbf{A}$.
$|\mathcal{A}|$ denotes the size of set $\mathcal{A}$.
$\mathsf{blkdiag}(\cdot)$ denotes a block-diagonal matrix.
$\mathsf{Tr}(\cdot)$ denotes the trace of a matrix.

\section{Dynamically Group-Connected BD-RIS Model}

\begin{figure}
    \centering
    \includegraphics[width = 0.47\textwidth]{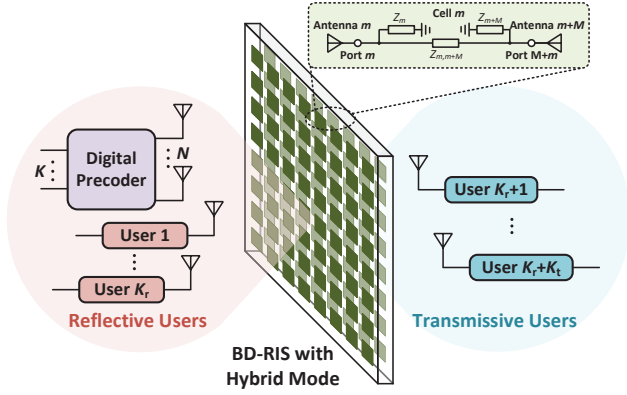}
    \caption{Diagram of BD-RIS supporting the hybrid mode.}\label{fig:sys_mod}\vspace{-0.3 cm}
\end{figure}

As illustrated in \cite{li2022generalized} and also in Fig. \ref{fig:RIS_tree}, the hybrid reflective and transmissive mode of the RIS is essentially based on the group-connected reconfigurable impedance network. More specifically, for a BD-RIS working on the hybrid mode, each two antenna ports are connected to each other, thereby constructing one ``cell'' as illustrated in Fig. \ref{fig:sys_mod}. Within each cell, two antennas are back to back placed such that each antenna covers half space.  
For an $M$-cell BD-RIS, we assume antenna ports $m$ and $M+m$ are connected to each other to support the hybrid mode. Mathematically, we define $\mathbf{\Phi}_\mathrm{r} \in \mathbb{C}^{M\times M}$ and $\mathbf{\Phi}_\mathrm{t} \in \mathbb{C}^{M\times M}$ as BD-RIS matrices in charge of users from two sides of the hybrid BD-RIS, which are sub-matrices of the scattering matrix $\mathbf{\Phi}\in\mathbb{C}^{2M\times 2M}$ for the $2M$-port reconfigurable impedance network\footnote{The scattering matrix $\mathbf{\Phi}\in\mathbb{C}^{2M\times 2M}$ for the $2M$-port reconfigurable impedance network is determined by the impedance matrix $\mathbf{Z}\in\mathbb{C}^{2M\times 2M}$ of the impedance network, that is $\mathbf{\Phi} = (\mathbf{Z}+Z_0\mathbf{I}_{2M})^{-1}(\mathbf{Z}-Z_0\mathbf{I}_{2M})$. $\mathbf{Z}$ can be constructed by impedance values $Z_{m,n}$ as detailed in \cite{shen2021modeling}.}, that is
\begin{equation}
    \mathbf{\Phi}_\mathrm{r} = [\mathbf{\Phi}]_{1:M,1:M}, ~ \mathbf{\Phi}_\mathrm{t} = [\mathbf{\Phi}]_{M+1:2M,1:M}.
\end{equation} 
They are constrained by \cite{li2022generalized}
\begin{equation}
    \mathbf{\Phi}_\mathrm{r}^H\mathbf{\Phi}_\mathrm{r} + \mathbf{\Phi}_\mathrm{t}^H\mathbf{\Phi}_\mathrm{t} = \mathbf{I}_M.
    \label{eq:RIS_constraint}
\end{equation}  
Constraint (\ref{eq:RIS_constraint}) is a general description without specifying the  mathematical characteristics of $\mathbf{\Phi}_\mathrm{r}$ and $\mathbf{\Phi}_\mathrm{t}$, both of which depend on the architectures of BD-RIS cells. In the following subsections, we will propose a novel and unified cell-wise dynamically group-connected (CW-DGC) architecture of BD-RIS based on a dynamic grouping strategy and derive the corresponding constraints of $\mathbf{\Phi}_\mathrm{r}$ and $\mathbf{\Phi}_\mathrm{t}$. Then we will discuss the practical implementation of the proposed architecture.

\vspace{-0.4 cm}

\subsection{Mathematical Model}
In this work, we propose a CW-DGC architecture of BD-RIS by adapting the grouping strategy to channel environments. 
Different from the cell-wise group-connected (CW-GC) architecture proposed in \cite{li2022generalized}, where $M$ cells are uniformly divided into $G$ subsets and every adjacent $\bar{M} = M/G$ cells are connected to each other, CW-DGC is based on a dynamic grouping strategy, which is embodied in two aspects: 1) The number of cells in each group can be different; 2) the positions/indexes of cells in the same group are not restricted to adjacent ones. To model this dynamic grouping strategy, we first present the following definition. 
\begin{defin}
    Define $G$ subsets $\mathcal{D}_1, \ldots, \mathcal{D}_G$ which store indexes of RIS cells for each group. These subsets satisfy the following constraints:
    \begin{subequations}
        \label{eq:dynamic_grouping}
        \begin{align}
            &\mathcal{D}_g \ne \varnothing, \forall g \in \mathcal{G} = \{1,\ldots,G\},\\
            &\mathcal{D}_p \cap \mathcal{D}_q = \varnothing, \forall p \ne q, p, q \in \mathcal{G},\\        
            &\cup_{g = 1}^G \mathcal{D}_g = \mathcal{M} = \{1,\ldots,M\}, 
        \end{align}
    \end{subequations}
    which indicate 1) that each group contains at lease one cell, 2) that each BD-RIS cell can only be mapped into one group, and 3) that there is no overlap among different groups\footnote{Herein we provide a simple example to explain Definition 1. For a 4-cell BD-RIS with 2 groups, 2 possible grouping results could be $\mathcal{D}_1 = \{1,3\}$, $\mathcal{D}_2 = \{2,4\}$ or $\mathcal{D}_1 = \{2\}$, $\mathcal{D}_2 = \{1,3,4\}$.}. 
\end{defin}

Based on Definition 1, we have the following corollary.  
\begin{corollary}
    The CW-DGC BD-RIS matrices, $\mathbf{\Phi}_\mathrm{r}$ and $\mathbf{\Phi}_\mathrm{t}$, should satisfy the following constraints\footnote{In this work, we consider the hybrid mode of BD-RIS to demo the proposed CW-DGC architecture. However, the CW-DGC architecture can also support multi-sector mode and the illustrated model can be easily extended to the multi-sector case.}
    \begin{subequations}
        \label{eq:dynamic_constraint}
        \begin{align}
            \label{eq:dynamic_constraint_a}
        &[\mathbf{\Phi}_\mathrm{t}]_{m,n} = 0, \forall m \in \mathcal{D}_p, \forall n \in \mathcal{D}_q, \forall p \ne q, p,q\in\mathcal{G},\\
        \label{eq:dynamic_constraint_b}
        &[\mathbf{\Phi}_\mathrm{r}]_{m,n} = 0, \forall m \in \mathcal{D}_p, \forall n \in \mathcal{D}_q, \forall p \ne q,p,q\in\mathcal{G},\\
        \label{eq:dynamic_constraint_c}
        &\mathbf{\Phi}_{\mathrm{t},\mathcal{D}_g}^H\mathbf{\Phi}_{\mathrm{t}, \mathcal{D}_g} + \mathbf{\Phi}_{\mathrm{r},\mathcal{D}_g}^H\mathbf{\Phi}_{\mathrm{r},\mathcal{D}_g} = \mathbf{I}_{|\mathcal{D}_g|}, \forall g\in\mathcal{G},
        \end{align}
    \end{subequations}
    where $\mathbf{\Phi}_{\mathrm{t/r},\mathcal{D}_g}$ is a sub-matrix of $\mathbf{\Phi}_{\mathrm{t/r}}$ which selects rows and columns of $\mathbf{\Phi}_{\mathrm{t/r}}$ according to indexes in the set $\mathcal{D}_g$,$\forall g\in\mathcal{G}$. 
\end{corollary}

\begin{proof}
    Please refer to the Appendix. 
\end{proof}

\begin{remark}
Constraint (\ref{eq:dynamic_constraint}) indicates that  $\mathbf{\Phi}_\mathrm{r}$ and $\mathbf{\Phi}_\mathrm{t}$ of the CW-DGC BD-RIS are essentially permuted block-diagonal matrices. Specifically, when the BD-RIS has a CW-GC architecture \cite{li2022generalized}, the $\bar{M}$ cells within a fixed group $\mathcal{G}_g = \{(g-1)\bar{M} + 1, \ldots,g\bar{M}\}$, $\forall g\in\mathcal{G}$ are connected to each other. In this case, we have $[\mathbf{\Phi}_\mathrm{r}]_{m,n} \ne 0$, $[\mathbf{\Phi}_\mathrm{t}]_{m,n} \ne 0$, $\forall m,n\in\mathcal{G}_g$, and the resulting $\mathbf{\Phi}_\mathrm{r}$ and $\mathbf{\Phi}_\mathrm{t}$ are both block diagonal matrices.
\end{remark}

\begin{figure}
    \centering
    \includegraphics[width = 0.48\textwidth]{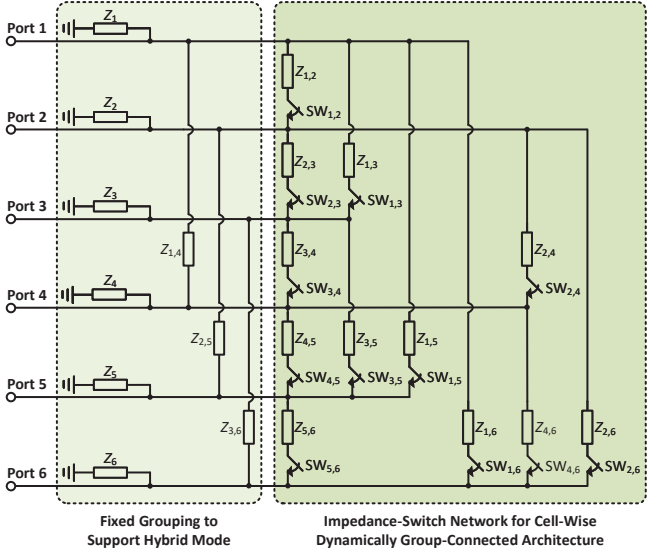}
    \caption{Example of a 3-cell BD-RIS with hybrid mode and CW-DGC architecture realized by an impedance-switch network.}\label{fig:Group_Stra}\vspace{-0.5 cm}
\end{figure}

\vspace{-0.3 cm}

\subsection{Implementation and Discussion}
To show the feasibility of the proposed CW-DGC architecture, we employ an impedance-switch network, which enables a joint control of impedance components and the ON/OFF states of switches. 
Specifically, impedance components can be realized by using tunable capacitance and inductance, such as using varactors. 
Fig. \ref{fig:Group_Stra} provides an example of the realization of CW-DGC BD-RIS with $M = 3$. Specifically, ports $m$ and $M+m$ are connected to each other by an impedance to support the hybrid mode (as shown in the left-hand side of Fig. \ref{fig:Group_Stra}), while ports $m$ and $m'$, $m'\ne M+m$, are connected by the series of one impedance and one switch to support the dynamic grouping (as shown in the right-hand side of Fig. \ref{fig:Group_Stra}). 
Therefore, to realize an $M$-cell BD-RIS with hybrid mode and the CW-DGC architecture, $M(2M+1)$ impedance components and $2M(M-1)$ switches are required. 
To realize the dynamic grouping, when cells $m$ and $n$, $m<n$, are in the same group, the $(m,n)$-th, $(m,M+n)$-th, $(n,M+m)$-th, and $(M+m,M+n)$-th switches are turned ON and the corresponding impedance components are tuned based on CSI.

\begin{remark}
    Although the number of connected links, that is the number of connections between every two BD-RIS antennas, of the impedance-switch network as illustrated in Fig. \ref{fig:Group_Stra} scales with $M^2$, the activated links, that is the actual number of connections between RIS antennas with given CSI and $G$, of the CW-DGC architecture could be much smaller than the total number of connected links. 
    Based on the above implementation, with a given number of groups $G$, the number of activated links for CW-DGC BD-RIS is $\sum_{g=1}^G|\mathcal{D}_g|(2|\mathcal{D}_g| + 1)$.
\end{remark}

\begin{remark}
    The impedance-switch network as illustrated in Section II-B is a general framework, which enables cell-wise fully/group/single-connected (CW-FC/GC/SC) architectures in \cite{li2022generalized}. Specifically, the CW-FC architecture is realized when all switches are turned ON, while the CW-SC architecture is realized when all switches are turned OFF. This finding indicates that different architectures of BD-RIS could share the same realization, which is preferable for the implementation of BD-RIS in practice.
\end{remark}

In the following section, we will apply the CW-DGC BD-RIS in a MU-MISO system and propose efficient algorithms to optimize the BD-RIS matrices with dynamic grouping strategy.

\section{Dynamically Group-Connected BD-RIS Design}

\subsection{BD-RIS-Aided MU-MISO with Dynamic Grouping Strategy}

To demonstrate the advantage of the proposed CW-DGC architecture, we apply the CW-DGC BD-RIS with the hybrid mode into a MU-MISO system as shown in Fig. \ref{fig:sys_mod}, where an $N$-antenna base station (BS) serves $K$ single-antenna users with the assistance of an $M$-cell BD-RIS. 
Among $K$ users, $K_\mathrm{r}$ users, namely reflective users, $\mathcal{K}_\mathrm{r} = \{1,\ldots,K_\mathrm{r}\}$, are located at one side of the BD-RIS, and $K_\mathrm{t} = K - K_\mathrm{r}$ users, namely transmissive users, $\mathcal{K}_\mathrm{t} = \{K_\mathrm{r} + 1,\ldots,K\}$, are located at the other side, $\mathcal{K} = \mathcal{K}_\mathrm{t}\cup \mathcal{K}_\mathrm{r}$. 
Define $\mathbf{s} \triangleq [s_1, \ldots, s_K]^T \in \mathbb{C}^{K}$ as the transmit symbol vector, $\mathbb{E}\{\mathbf{s}\mathbf{s}^H\} = \mathbf{I}_{K}$, and $\mathbf{W} \triangleq [\mathbf{w}_{1}, \ldots, \mathbf{w}_{K}] \in \mathbb{C}^{N\times K}$ as the transmit precoder. The received signal at the user side is
\begin{equation}
    \label{eq:receive_signal}
    \begin{aligned}
        y_{k} = \tilde{\mathbf{h}}_{k}^H\mathbf{w}_{k}s_{k} + \tilde{\mathbf{h}}_{k}^H\sum_{p \in \mathcal{K}, p \ne k}\mathbf{w}_{p}s_{p}+ n_{k}, \forall k \in \mathcal{K},
    \end{aligned}
\end{equation}
where $\tilde{\mathbf{h}}_{k} = (\mathbf{h}_{k}^H\mathbf{\Phi}_i \mathbf{G})^H$, $\forall k \in\mathcal{K}_i$, $\forall i \in\{\mathrm{t,r}\}$, $\mathbf{h}_{k} \in \mathbb{C}^{M}$, $\forall k \in \mathcal{K}$ is the channel vector between the RIS and each user, $\mathbf{G} \in \mathbb{C}^{M\times N}$ is the channel matrix between the BS and the RIS, and $n_{k} \sim \mathcal{CN}(0, \sigma_{k}^2)$, $\forall k \in \mathcal{K}$ is the noise.

Aiming at jointly optimizing the precoder and BD-RIS matrix to maximize the sum-rate, we have
\begin{subequations}
    \label{eq:problem0}
    \begin{align}
    \label{eq:obj0}
    \max_{\mathbf{W}, \mathbf{\Phi}_\mathrm{t}, \mathbf{\Phi}_\mathrm{r}} ~~ &\sum_{k\in\mathcal{K}} \log_2\left(1 + \frac{|\tilde{\mathbf{h}}_{k}^H \mathbf{w}_{k}|^2}{
        \sum_{\substack{p \in \mathcal{K} \\ p\ne k}}|\tilde{\mathbf{h}}_{k}^H \mathbf{w}_{p}|^2 + \sigma_{k}^2}\right)\\
    \label{eq:p0_b}
    \mathrm{s.t.} ~~~~&\mathcal{D}_p \cap \mathcal{D}_q = \varnothing, \forall p \ne q,\\
    \label{eq:p0_c}
    &\mathcal{D}_g \ne \varnothing, \forall g,\\   
    \label{eq:p0_d}
    &\cup_{g = 1}^G \mathcal{D}_g = \mathcal{M}, \\
    \label{eq:p0_e}
    &[\mathbf{\Phi}_\mathrm{t}]_{m,n} = 0, \forall m \in \mathcal{D}_p, \forall n \in \mathcal{D}_q, p \ne q,\\
    \label{eq:p0_e1}
    &[\mathbf{\Phi}_\mathrm{r}]_{m,n} = 0, \forall m \in \mathcal{D}_p, \forall n \in \mathcal{D}_q, p \ne q,\\
    \label{eq:p0_f}
    &\mathbf{\Phi}_{\mathrm{t},\mathcal{D}_g}^H\mathbf{\Phi}_{\mathrm{t}, \mathcal{D}_g} + \mathbf{\Phi}_{\mathrm{r},\mathcal{D}_g}^H\mathbf{\Phi}_{\mathrm{r},\mathcal{D}_g} = \mathbf{I}_{|\mathcal{D}_g|}, \forall g,\\
    \label{eq:p0_g}
    &\|\mathbf{W}\|_F^2 \le P,
    \end{align}
\end{subequations}
where $P$ is the total transmit power at the BS.

\subsection{CW-DGC BD-RIS Design}

Problem (\ref{eq:problem0}) is a joint transmit precoder and BD-RIS design with complicated objective function and constraints. To simplify the design, we transform problem (\ref{eq:problem0}) into a more tractable multi-block problem based on fractional programming \cite{shen2018fractional}, which 1) takes the ratio parts of the objective out of the $\log(\cdot)$ functions and 2) transforms the ratio terms into integral expressions by introducing auxiliary vectors ${\bm \iota} \triangleq [\iota_1, \ldots, \iota_K]^T \in \mathbb{R}^K$ and ${\bm \tau} \triangleq [\tau_{1}, \ldots, \tau_{K}]^T \in \mathbb{C}^K$, yielding
\begin{subequations}
    \label{eq:problem1}
    \begin{align}
        \non
        \max_{\substack{\mathbf{W}, \mathbf{\Phi}_\mathrm{t}, \mathbf{\Phi}_\mathrm{r}\\{\bm \iota}, {\bm \tau}}} &\sum_{k\in\mathcal{K}} \Big(\log_2(1 + \iota_{k}) - \iota_{k} + 
        2\Re\{\tilde{\tau}_{k}^*\tilde{\mathbf{h}}_{k}^H\mathbf{w}_{k}\} \\
        \label{eq:p1_a}
        &- \sum_{p\in\mathcal{K}} |\tau_{k}^*\tilde{\mathbf{h}}_{k}^H\mathbf{w}_{p}|^2 - |\tau_{k}|^2\sigma_{k}^2\Big)\\
        \label{eq:p1_b}
        \mathrm{s.t.} ~~ & \textrm{(\ref{eq:p0_b})-(\ref{eq:p0_g})},
    \end{align}
\end{subequations}
where $\tilde{\tau}_{k} = \sqrt{1 + \iota_{k}}\tau_{k}$, $\forall k \in \mathcal{K}$.
Then problem (\ref{eq:problem1}) can be efficiently solved by iteratively updating four blocks, i.e., ${\bm \iota}$, ${\bm \tau}$, $\mathbf{W}$, and $\{\mathbf{\Phi}_\mathrm{t}, \mathbf{\Phi}_\mathrm{r}\}$, until convergence. 
Sub-problems regarding blocks ${\bm \iota}$, ${\bm \tau}$, and $\mathbf{W}$ are all convex optimizations whose closed-form solutions can be easily obtained, while the design of block $\{\mathbf{\Phi}_\mathrm{t}, \mathbf{\Phi}_\mathrm{r}\}$ is challenging due to the newly introduced constraints (\ref{eq:p0_b})-(\ref{eq:p0_f}). 
Therefore, in this paper we omit the details for updating the above three blocks due to the space limitation and focus on the design of block $\{\mathbf{\Phi}_\mathrm{t}, \mathbf{\Phi}_\mathrm{r}\}$. 

When $\mathbf{W}$, ${\bm \tau}$, and ${\bm \iota}$ are fixed, the first two and the last terms in objective (\ref{eq:p1_a}) are constants and can be removed. 
Therefore, the objective function with respect to $\mathbf{\Phi}_\mathrm{t}$ and $\mathbf{\Phi}_\mathrm{r}$ is
\begin{subequations}\label{eq:sub_phi}
    \begin{align}           
        &\sum_{k\in\mathcal{K}}\Big(\sum_{p\in\mathcal{K}} |\tau_{k}^*\mathbf{h}_{k}^H\mathbf{\Phi}_i\mathbf{G}\mathbf{w}_{p}|^2 - 2\Re\{\tilde{\tau}_{k}^*\mathbf{h}_{k}^H\mathbf{\Phi}_i\underbrace{\mathbf{G}\mathbf{w}_{k}}_{=\mathbf{g}_{k}}\}\Big) \\
        \non
        = &\sum_{i\in\{\mathrm{t,r}\}}\Big(\mathsf{Tr}(\mathbf{\Phi}_i\underbrace{\sum_{p\in\mathcal{K}} \mathbf{g}_p\mathbf{g}_p^H}_{=\mathbf{Y}}\mathbf{\Phi}_i^H\underbrace{\sum_{k\in\mathcal{K}_i}|\tau_k|^2\mathbf{h}_k\mathbf{h}_k^H}_{=\mathbf{Z}_i})\\
        \label{eq:sub_phi_b}
        &~~~- 2\Re\{\mathsf{Tr}(\mathbf{\Phi}_i\underbrace{\sum_{k\in\mathcal{K}_i}\mathbf{g}_k\mathbf{h}_k^H\tilde{\tau}_k^*}_{=\mathbf{X}_i})\} \Big)\\
        \non
        \overset{\textrm{(a)}}{=} &\sum_{i\in\{\mathrm{t,r}\}}\sum_{g\in\mathcal{G}} \Big(\mathsf{Tr}(\mathbf{\Phi}_{i,\mathcal{D}_g}\sum_{p\in\mathcal{G}} \mathbf{Y}_{\mathcal{D}_g, \mathcal{D}_p}\mathbf{\Phi}_{i,\mathcal{D}_p}^H \mathbf{Z}_{i,\mathcal{D}_p, \mathcal{D}_g})\\ 
        \label{eq:sub_phi_c}
        &~~~-2\Re\{\mathsf{Tr}(\mathbf{\Phi}_{i,\mathcal{D}_g}\mathbf{X}_{i,\mathcal{D}_g})\}\Big)\\
        \non
        \overset{\textrm{(b)}}{\approx} &\sum_{g\in\mathcal{G}}\sum_{i\in\{\mathrm{t,r}\}}\Big(\mathsf{Tr}(\mathbf{\Phi}_{i,\mathcal{D}_g} \mathbf{Y}_{\mathcal{D}_g, \mathcal{D}_g}\mathbf{\Phi}_{i,\mathcal{D}_g}^H\mathbf{Z}_{i,\mathcal{D}_g, \mathcal{D}_g})\\
        &~~~- 2\Re\{\mathsf{Tr} (\mathbf{\Phi}_{i,\mathcal{D}_g}\mathbf{X}_{i,\mathcal{D}_g})\}\Big)\\
        \label{eq:sub_phi_d}
        = &\sum_{g\in\mathcal{G}} f_g(\mathbf{\Phi}_{\mathrm{t},\mathcal{D}_g}, \mathbf{\Phi}_{\mathrm{r},\mathcal{D}_g}),
    \end{align}
\end{subequations}
where (a) holds by defining matrices $\mathbf{X}_{i,\mathcal{D}_g}$, $\mathbf{Y}_{\mathcal{D}_p, \mathcal{D}_g}$ and $\mathbf{Z}_{i,\mathcal{D}_p,\mathcal{D}_g}$, $\forall g\in\mathcal{G}$, $\forall i\in\{\mathrm{t,r}\}$, which are sub-matrices selecting rows and columns according to $\mathcal{D}_g$ from $\mathbf{X}_i$, and selecting rows according to $\mathcal{D}_p$ and columns according to $\mathcal{D}_g$ from $\mathbf{Y}$ and $\mathbf{Z}_i$, $\forall i\in\{\mathrm{t,r}\}$, respectively; (b) holds by the finding (based on large amount of simulations) that the value of the term $\sum_{p \ne g}\mathbf{Y}_{\mathcal{D}_g, \mathcal{D}_p}\mathbf{\Phi}_{i,\mathcal{D}_p}^H\mathbf{Z}_{i,\mathcal{D}_p, \mathcal{D}_g}$ is negligible compared to $\mathbf{X}_{i,\mathcal{D}_g}$ due to the severe double fading of the cascaded BS-RIS-user channels.
Based on the above derivations, the sub-problem with respect to $\mathbf{\Phi}_\mathrm{t}$ and $\mathbf{\Phi}_\mathrm{r}$ is
\begin{equation}\label{eq:sub_phi1}
        \min_{\substack{\mathcal{D}_g, \forall g\in\mathcal{G}\\ \mathbf{\Phi}_{\mathrm{t},\mathcal{D}_g}, \mathbf{\Phi}_{\mathrm{r},\mathcal{D}_g}, \forall g\in\mathcal{G}}} \sum_{g\in\mathcal{G}}f_g(\mathbf{\Phi}_{\mathrm{t},\mathcal{D}_g}, \mathbf{\Phi}_{\mathrm{r},\mathcal{D}_g})
        ~~\text{s.t.} ~\textrm{(\ref{eq:p0_b})-(\ref{eq:p0_f})}.
\end{equation}
Problem (\ref{eq:sub_phi1}) is still difficult to solve due to the coupling of grouping $\mathcal{D}_1, \ldots, \mathcal{D}_G$ and non-zero parts of the BD-RIS matrices $\mathbf{\Phi}_{\mathrm{t/r},\mathcal{D}_g}$, $\forall g\in\mathcal{G}$. To further simplify the design, we propose to iteratively\footnote{Note that the initialization of the dynamic grouping subsets, the BD-RIS matrices and the transmit precoder is required. For simplicity, we initialize dynamic grouping sets as $\mathcal{D}_g = \mathcal{G}_g$, $\forall g\in\mathcal{G}$, the BD-RIS matrices as diagonal ones with non-zero entries having constant magnitude $\frac{1}{\sqrt{2}}$ and random phase, and the transmit precoder as zero-forcing precoder.} update $\mathcal{D}_g$, $\forall g\in\mathcal{G}$ and $\mathbf{\Phi}_{\mathrm{t/r},\mathcal{D}_g}$, $\forall g\in\mathcal{G}$, until the convergence is achieved.

\subsubsection{Dynamic Grouping Optimization} 
When RIS matrices are fixed, the grouping problem is 
\begin{equation}
    \label{eq:sub_dynamic_grouping_objective}
        \min_{\mathcal{D}_g, \forall g\in\mathcal{G}} ~ \sum_{g\in\mathcal{G}} f_g(\mathbf{\Phi}_{\mathrm{t},\mathcal{D}_g}, \mathbf{\Phi}_{\mathrm{r},\mathcal{D}_g})~~
        \mathrm{s.t.} ~\textrm{(\ref{eq:p0_b})-(\ref{eq:p0_d})}.
\end{equation}
We propose a simple yet practical iterative algorithm to solve the grouping problem. Given a proper initialization of $\mathcal{D}_1, \ldots, \mathcal{D}_G$, we aim to successively determine which group each RIS cell should belong to based on the following steps. 

\textit{Step 1:} For cell $m$, $\forall m\in\mathcal{M}$, we first check which group this cell belongs to and tag the exact position as $g_\mathrm{tag}$. To guarantee each group contains at least one BD-RIS cell, we check the size of $\mathcal{D}_{g_\mathrm{tag}}$. When $|\mathcal{D}_{g_\mathrm{tag}}| > 1$, the following Steps 2-4 are executed. 

\textit{Step 2:} We calculate the value of objective (\ref{eq:sub_dynamic_grouping_objective}), which is marked as $\bar{f}_{m,g}$, when cell $m$ is either in group $g_\mathrm{tag}$, i.e., $g=g_\mathrm{tag}$ or moved to other $G-1$ groups, i.e., $g\ne g_\mathrm{tag}$, according to (\ref{eq:bar_f}):
\begin{equation}
\label{eq:bar_f}
    \begin{aligned}
        \bar{f}_{m,g} = \begin{cases}
            \sum_{p\in\mathcal{G}} f_p(\mathbf{\Phi}_{\mathrm{t},\mathcal{D}_p}, \mathbf{\Phi}_{\mathrm{r},\mathcal{D}_p}), &g = g_\mathrm{tag},\\
            f_{g_\mathrm{tag}}(\mathbf{\Phi}_{\mathrm{t},\widetilde{\mathcal{D}}_{g_\mathrm{tag}}^m}, \mathbf{\Phi}_{\mathrm{r},\widetilde{\mathcal{D}}_{g_\mathrm{tag}}^m})\\
              ~~+ f_g(\mathbf{\Phi}_{\mathrm{t},\widehat{\mathcal{D}}_g^m}, \mathbf{\Phi}_{\mathrm{r},\widehat{\mathcal{D}}_g^m}) \\
              ~~+ \sum_{\substack{p\in\mathcal{G}\\ p\ne g\\p\ne g_\mathrm{tag}}} f_p(\mathbf{\Phi}_{\mathrm{t},\mathcal{D}_p}, \mathbf{\Phi}_{\mathrm{r},\mathcal{D}_p}), &g \ne g_\mathrm{tag},
        \end{cases}
    \end{aligned}
\end{equation}
where $\widetilde{\mathcal{D}} _{g_\mathrm{tag}}^m = \mathcal{D}_{g_\mathrm{tag}}\setminus\{m\}$, $\widehat{\mathcal{D}}_g^m = \mathcal{D}_g\cup\{m\}$.

\textit{Step 3:} Among all $\bar{f}_{m,g}$, $\forall g\in\mathcal{G}$, we find the index with the minimum value, i.e.,
\begin{equation}
    \label{eq:find_g}
    g_m^\star = \arg\min_{g} \bar{f}_{m,g}.
\end{equation}

\textit{Step 4:} When $g_m^\star \ne g_\mathrm{tag}$, we update the corresponding dynamic subsets as
\begin{equation}
    \label{eq:update_D}
    \mathcal{D}_{g_m^\star} = \mathcal{D}_{g_m^\star}\cup\{m\}, ~
    \mathcal{D}_{g_\mathrm{tag}} = \mathcal{D}_{g_\mathrm{tag}}\setminus\{m\}.
\end{equation}

Successively updating positions of all $M$
RIS cells, we can obtain the dynamic grouping subsets $\mathcal{D}_1, \ldots, \mathcal{D}_G$.
The computational complexity of the proposed dynamic grouping optimization mainly comes from step 2, which requires around $\mathcal{O}(\sum_{g=1}^G|\mathcal{D}_g|^{2.37})$ due to the matrix multiplication, yielding a worst-case total complexity $\mathcal{O}(M\sum_{g=1}^G|\mathcal{D}_g|^{2.37})$.

\subsubsection{BD-RIS Matrix Optimization} With fixed grouping subsets, the BD-RIS matrix design problem for each group is separated, which is given by
\begin{subequations}
    \label{eq:sub_dynamic_phi_g2}
    \begin{align}
        \min_{\mathbf{\Phi}_{\mathcal{D}_g}} ~ &\mathsf{Tr}(\mathbf{\Phi}_{\mathcal{D}_g} \mathbf{Y}_{\mathcal{D}_g, \mathcal{D}_g}\mathbf{\Phi}_{\mathcal{D}_g}^H\mathbf{Z}_{\mathcal{D}_g})
        - 2\Re\{\mathsf{Tr}(\mathbf{\Phi}_{\mathcal{D}_g}\mathbf{X}_{\mathcal{D}_g})\}\\
        \mathrm{s.t.} ~~ & \mathbf{\Phi}_{\mathcal{D}_g}^H\mathbf{\Phi}_{\mathcal{D}_g} = \mathbf{I}_{|\mathcal{D}_g|},
    \end{align}
\end{subequations}
where the fresh notations are respectively defined as $\mathbf{\Phi}_{\mathcal{D}_g} \triangleq [\mathbf{\Phi}_{\mathrm{t},\mathcal{D}_g}^H, \mathbf{\Phi}_{\mathrm{r},\mathcal{D}_g}^H]^H \in \mathbb{C}^{2|\mathcal{D}_g|\times |\mathcal{D}_g|}$, $\mathbf{X}_{\mathcal{D}_g} \triangleq [\mathbf{X}_{\mathrm{t},\mathcal{D}_g},\mathbf{X}_{\mathrm{r},\mathcal{D}_g}] \in \mathbb{C}^{|\mathcal{D}_g|\times 2|\mathcal{D}_g|}$, and $\mathbf{Z}_{\mathcal{D}_g} \triangleq \mathsf{blkdiag} (\mathbf{Z}_{\mathrm{t},\mathcal{D}_g,\mathcal{D}_g}, \mathbf{Z}_{\mathrm{r},\mathcal{D}_g,\mathcal{D}_g}) \in \mathbb{C}^{2|\mathcal{D}_g|\times 2|\mathcal{D}_g|}$.
Problem (\ref{eq:sub_dynamic_phi_g2}) is an unconstrained optimization on the Stiefel manifold and can be solved by typical conjugate-gradient methods on the manifold space \cite{absil2009optimization} with complexity $\mathcal{O}(I_{\mathrm{cg},g}|\mathcal{D}_g|^3)$, where $I_{\mathrm{cg},g}$, $\forall g\in\mathcal{G}$ denotes the number of iterations. 
After solving (\ref{eq:sub_dynamic_phi_g2}), BD-RIS matrices for each group are given by
\begin{equation}
    \label{eq:opt_phi_dynamic_g}
    \begin{aligned}
    &\mathbf{\Phi}_{\mathrm{t}, \mathcal{D}_g}^\star = [\mathbf{\Phi}_{\mathcal{D}_g}^\star]_{1:|\mathcal{D}_g^\star|,:}, \forall g \in \mathcal{G},\\
    &\mathbf{\Phi}_{\mathrm{r}, \mathcal{D}_g}^\star = [\mathbf{\Phi}_{\mathcal{D}_g}^\star]_{|\mathcal{D}_g^\star| + 1 : 2|\mathcal{D}_g^\star|,:}, \forall g \in \mathcal{G}.
    \end{aligned}
\end{equation}
With $\mathbf{\Phi}_{\mathrm{t}, \mathcal{D}_g}^\star$ and $\mathbf{\Phi}_{\mathrm{r}, \mathcal{D}_g}^\star$, $\forall g \in \mathcal{G}$, we need to perform a restoring operation to get the original $\mathbf{\Phi}_\mathrm{t}^\star$ and $\mathbf{\Phi}_\mathrm{r}^\star$ satisfying (\ref{eq:p0_e}) and (\ref{eq:p0_e1}). To this end, we introduce $G$ vectors $\mathbf{d}_g \in \mathbb{R}^{|\mathcal{D}_g|}$, $\forall g\in\mathcal{G}$, which list the indexes in $\mathcal{D}_g$ in ascending order. Then we have
\begin{equation}
    \label{eq:reduction}
    \begin{aligned}
    &[\mathbf{\Phi}_{\mathrm{t}}^\star]_{\mathbf{d}_g^T,\mathbf{d}_g^T} = \mathbf{\Phi}_{\mathrm{t}, \mathcal{D}_g}^\star,
    [\mathbf{\Phi}_{\mathrm{r}}^\star]_{\mathbf{d}_g^T,\mathbf{d}_g^T} = \mathbf{\Phi}_{\mathrm{r}, \mathcal{D}_g}^\star, \forall g \in \mathcal{G}.
    \end{aligned}
\end{equation}

\subsubsection{Summary} After iteratively updating grouping subsets and BD-RIS matrices until some convergence thresholds are achieved, we can get $\mathcal{D}_1^\star, \ldots, \mathcal{D}_G^\star$ and corresponding BD-RIS matrices.
For clarity, the CW-DGC BD-RIS design is summarized in Algorithm \ref{alg:dynamic_group}.
The complexity of Algorithm \ref{alg:dynamic_group} mainly comes from the BD-RIS matrix optimization, yielding a total complexity $\mathcal{O}(I\sum_{g=1}^GI_{\mathrm{cg},g}|\mathcal{D}_g|^3)$, where $I$ denotes the number of iterations of Algorithm \ref{alg:dynamic_group}. 
Specifically, when the BD-RIS has a CW-GC architecture, there is no need to do iterations and the complexity reduces to $\mathcal{O}(\sum_{g=1}^GI_{\mathrm{cg},g}\bar{M}^3)$.

\begin{algorithm}[t]
    \caption{CW-DGC BD-RIS Design}
    \label{alg:dynamic_group}
    \begin{algorithmic}[1]
        \REQUIRE $\mathbf{h}_{k}$, $\forall k \in \mathcal{K}$, $\mathbf{G}$, ${\bm \iota}$, ${\bm \tau}$, $\mathbf{W}$, $\mathbf{\Phi}_\mathrm{t}$, $\mathbf{\Phi}_\mathrm{r}$, $\mathcal{D}_1, \ldots, \mathcal{D}_G$.
        \ENSURE $\mathbf{\Phi}_\mathrm{t}^\star$, $\mathbf{\Phi}_\mathrm{r}^\star$.
            \STATE {Calculate $\mathbf{X}_\mathrm{t/r}$, $\mathbf{Y}$, $\mathbf{Z}_\mathrm{t/r}$.}
            \WHILE {no convergence of objective (\ref{eq:sub_phi_d}) }
                \FOR {$m = 1 : M$}
                    \STATE {Find $m \in \mathcal{D}_g$ and set $g_\mathrm{tag} = g$.}
                    \IF {$|\mathcal{D}_g| > 1$}
                        \STATE {Calculate $\bar{f}_{m,g}$, $\forall g$ by (\ref{eq:bar_f}).}
                        \STATE {Find $g_m^\star$ by solving $g_m^\star = \arg\min_{g} \bar{f}_{m,g}$.}
                        \IF {$g_m^\star \ne g_\mathrm{tag}$}
                            \STATE {Update $\mathcal{D}_{g_m^\star} = \mathcal{D}_{g_m^\star}\cup\{m\}$.} 
                            \STATE {Update $\mathcal{D}_{g_\mathrm{tag}} = \mathcal{D}_{g_\mathrm{tag}}\setminus\{m\}$.}
                        \ENDIF
                    \ENDIF
                \ENDFOR
                \STATE {Set $\mathbf{\Phi}_\mathrm{t} = \mathbf{0}$, $\mathbf{\Phi}_\mathrm{r} = \mathbf{0}$.}
                \FOR {$g = 1 : G$}
                    \STATE {Split $\mathbf{X}_{\mathcal{D}_g}$, $\mathbf{Y}_{\mathcal{D}_g, \mathcal{D}_g}$, $\mathbf{Z}_{\mathcal{D}_g}$.}
                    \STATE {Solve problem (\ref{eq:sub_dynamic_phi_g2}) by the manifold method \cite{absil2009optimization}.}
                    \STATE {Obtain $\mathbf{\Phi}_{\mathrm{t}, \mathcal{D}_g}^\star$ and $\mathbf{\Phi}_{\mathrm{r}, \mathcal{D}_g}^\star$ by (\ref{eq:opt_phi_dynamic_g}).}
                    \STATE {Restore $\mathbf{\Phi}_\mathrm{t}^\star$ and $\mathbf{\Phi}_\mathrm{r}^\star$ by (\ref{eq:reduction}).}
                \ENDFOR
            \ENDWHILE
        \STATE {Return $\mathbf{\Phi}_\mathrm{t}^\star$, $\mathbf{\Phi}_\mathrm{r}^\star$.}
    \end{algorithmic}
\end{algorithm}

\section{Performance Evaluation}

In this section, we present simulation results to demonstrate the performance of the proposed design. 
In the following simulations, we assume Rician fading channels for the BS-RIS link and Rayleigh fading channels for RIS-user links. In addition, we adopt the typical distance-dependent pathloss model to account for the large-scale fading for BS-RIS link, i.e.,  $\zeta_\mathrm{BI} = \zeta_0(\frac{d_\mathrm{BI}}{d_0})^{-\varepsilon_\mathrm{BI}}$ and RIS-user link, i.e., $\zeta_\mathrm{IU} = \zeta_0(\frac{d_\mathrm{IU}}{d_0})^{-\varepsilon_\mathrm{IU}}$, where $\zeta_0 = -30$ dB is the signal attenuation at a reference distance $d_0 = 1$ m. $d_\mathrm{BI} = 100$ m and $d_\mathrm{IU} = 10$ m are distances between the BS and RIS, and between the RIS and users, respectively. $\varepsilon_\mathrm{BI}$ and $\varepsilon_\mathrm{IU}$ are path loss exponents for BS-RIS and RIS-user links, respectively. The noise power is set to $\sigma^2_k = -80$ dBm, $\forall k\in\mathcal{K}$.

Fig. \ref{fig:grouping_comparison} shows a simulation result when BD-RIS has 36 cells ($6\times 6$ square). Three fixed grouping strategies, namely ``Horizontal'', ``Vertical'', and ``Interlaced'', are used in the simulation as illustrated in \ref{fig:fixed_grouping}, where cells with the same color belong to the same group. 
For comparison, we also add the performance of CW-FC and CW-SC architectures as upper- and lower-bounds, respectively \cite{li2022generalized}.
We plot sum-rate performance achieved by different architectures of BD-RIS in Fig. \ref{fig:SR_P}. 
It can be shown that our proposed dynamic grouping strategy outperforms three CW-GC cases. 
In addition, there are no obvious performance gaps among three CW-GC architectures. 
These two observations can be explained as follows. 1) With Rayleigh fading channels between the RIS and users and fixed group size, the positional information for users is not a decisive factor. 2) The change of group size, however, boosts the performance by adapting to the strength of different channel coefficients. Therefore, we can deduce from Fig. \ref{fig:grouping_comparison} that the performance of CW-GC BD-RIS in this case depends more on the change of group size than on locations of cells.

\begin{figure}
    \centering
    \subfigure[Three fixed grouping strategies with $G = 18$ and $G = 12$]{
    \label{fig:fixed_grouping}
    \includegraphics[width = 0.46\textwidth]{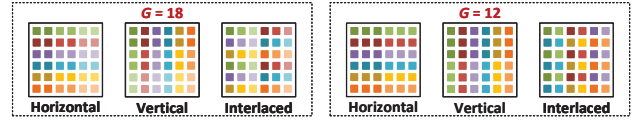}}
    \subfigure[Sum-rate versus transmit power $P$. Left: $G=18$; right: $G = 12$.]{
    \label{fig:SR_P}
    \includegraphics[width = 0.455\textwidth]{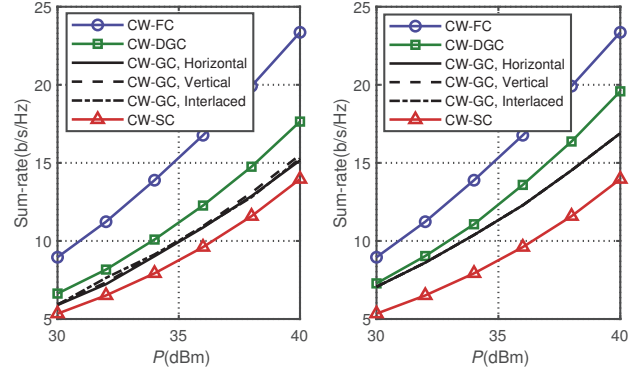}}
    \caption{Comparison among different grouping strategies ($M = 36$, $N = K = 6$, $K_\mathrm{t} = K_\mathrm{r} = 3$).}
    \label{fig:grouping_comparison}\vspace{-0.1 cm}
\end{figure}

In Fig. \ref{fig:SR_G}\footnote{Given that three fixed grouping strategies have similar performance, in Fig. \ref{fig:SR_G} we only plot the horizontal case, which is marked as ``CW-GC''.}, we show the sum-rate performance as a function of the number of groups $G$. 
With increasing number of groups, equivalently the decrease of the number of non-zero elements of the BD-RIS matrix, the performance gap between CW-DGC and CW-GC architectures becomes larger. 
For example, the performance achieved by a 36-cell CW-DGC architecture improves by 12\% when $G = 12$ while by 15\% with $G = 16$; the performance of a 64-cell CW-DGC architecture improves by 13\% when $G=16$ and by 21\% when $G=32$.  
This fact demonstrates that the benefit of dynamic grouping is much more obvious when the size of each group is relatively small.

Finally, Fig. \ref{fig:SR_M} shows the sum-rate versus the number of cells. We can observe that the proposed CW-DGC architecture always outperforms the CW-GC architecture. More importantly, with relatively large number of cells, e.g., $M=100$ in Fig. \ref{fig:SR_M}, the CW-DGC architecture can achieve performance close to the CW-FC case even when $G=25$. This observation shows the benefit of CW-DGC architecture in achieving satisfactory performance with a reduced control complexity.

\begin{figure}
    \centering
    \includegraphics[width = 0.48\textwidth]{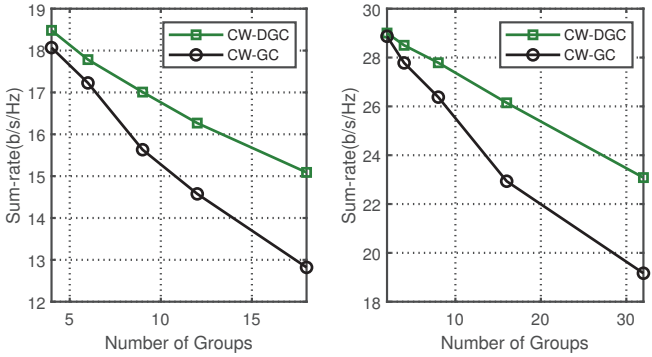}
    \caption{Sum-rate versus the number of groups $G$. Left:  $M = 36$; right:  $M = 64$. ($N = K = 6$, $K_\mathrm{t} = K_\mathrm{r} = 3$, $P = 38$ dBm)}
    \label{fig:SR_G}\vspace{-0.1 cm}
\end{figure}

\begin{figure}
    \centering
    \includegraphics[width = 0.475\textwidth]{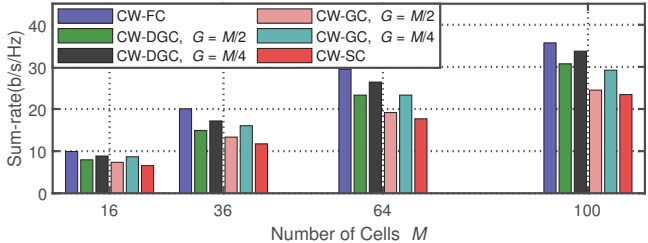}
    \caption{Sum-rate versus the number of cells $M$. ($N = K = 6$, $K_\mathrm{t} = K_\mathrm{r} = 3$, $P = 38$ dBm)}
    \label{fig:SR_M}\vspace{-0.1 cm}
\end{figure}

\section{Conclusions}
\label{sc:Conclusions}

In this paper, we propose a novel CW-DGC architecture for BD-RIS with the hybrid mode. To show the advantage of the proposed architecture, we deploy the CW-DGC BD-RIS into a MU-MISO communication system, and consider the joint dynamic grouping and BD-RIS matrix design to maximize the sum-rate performance. To efficiently solve the problem, we propose a simple yet practical algorithm to iteratively update dynamic grouping and BD-RIS matrix. Finally, simulation results demonstrate that the proposed CW-DGC architecture outperforms fixed CW-GC architectures. In addition, with Rayleigh fading channels between the RIS and users, the change of group size has a greater impact than that of cell locations on the achievable performance.
In the future, there are many interesting topics worth studying, such as more efficient beamforming design and hardware impairments of BD-RIS implementation.

\begin{appendix}[Proof of Corollary 1]
    Based on the derivations in \cite{li2022generalized}, since we consider the cell connection topology, while the antenna ports within the same cell are always connected to each other to support the hybrid mode, $\mathbf{\Phi}_\mathrm{r}$ and $\mathbf{\Phi}_\mathrm{t}$ share the same mathematical characteristics. Meanwhile, based on the circuit topology analysis in \cite{shen2021modeling}, when cell $m$ and cell $n$, $m\ne n$, are connected to each other by reconfigurable impedance components, the corresponding $(m,m)$-th, $(m,n)$-th, $(n,m)$-th, and $(n,n)$-th entries of $\mathbf{\Phi}_\mathrm{r}$ and $\mathbf{\Phi}_\mathrm{t}$ are nonzero; when cell $m$ and cell $n$ are not connected to each other, the corresponding $(m,n)$-th and $(n,m)$-th entries are zeros. Combining the above statements and Definition 1, we prove constraints (\ref{eq:dynamic_constraint_a}) and (\ref{eq:dynamic_constraint_b}). Constraint (\ref{eq:dynamic_constraint_c}) is straightforward given constraint (\ref{eq:RIS_constraint}) and that cells in the same group essentially construct a cell-wise fully-connected architecture with a reduced dimension. 
\end{appendix}

\bibliographystyle{IEEEtran}
\bibliography{references}

\begin{thebibliography}{10}
\providecommand{\url}[1]{#1}
\csname url@samestyle\endcsname
\providecommand{\newblock}{\relax}
\providecommand{\bibinfo}[2]{#2}
\providecommand{\BIBentrySTDinterwordspacing}{\spaceskip=0pt\relax}
\providecommand{\BIBentryALTinterwordstretchfactor}{4}
\providecommand{\BIBentryALTinterwordspacing}{\spaceskip=\fontdimen2\font plus
\BIBentryALTinterwordstretchfactor\fontdimen3\font minus
  \fontdimen4\font\relax}
\providecommand{\BIBforeignlanguage}[2]{{%
\expandafter\ifx\csname l@#1\endcsname\relax
\typeout{** WARNING: IEEEtran.bst: No hyphenation pattern has been}%
\typeout{** loaded for the language `#1'. Using the pattern for}%
\typeout{** the default language instead.}%
\else
\language=\csname l@#1\endcsname
\fi
#2}}
\providecommand{\BIBdecl}{\relax}
\BIBdecl

\bibitem{wu2019towards}
Q.~Wu and R.~Zhang, ``Towards smart and reconfigurable environment: Intelligent
  reflecting surface aided wireless network,'' \emph{IEEE Commun. Mag.},
  vol.~58, no.~1, pp. 106--112, 2019.

\bibitem{elmossallamy2020reconfigurable}
M.~A. ElMossallamy, H.~Zhang, L.~Song, K.~G. Seddik, Z.~Han, and G.~Y. Li,
  ``Reconfigurable intelligent surfaces for wireless communications:
  Principles, challenges, and opportunities,'' \emph{IEEE Trans. Cogn. Commun.
  Network.}, vol.~6, no.~3, pp. 990--1002, 2020.

\bibitem{shen2021modeling}
S.~Shen, B.~Clerckx, and R.~Murch, ``Modeling and architecture design of
  reconfigurable intelligent surfaces using scattering parameter network
  analysis,'' \emph{IEEE Trans. Wireless Commun.}, vol.~21, no.~2, pp.
  1229--1243, 2022.

\bibitem{nerini2022optimal}
M.~Nerini, S.~Shen, and B.~Clerckx, ``Closed-form global optimization of beyond
  diagonal reconfigurable intelligent surfaces,'' \emph{IEEE Trans. Wireless
  Commun.}, accepted, 2023.

\bibitem{nerini2021reconfigurable}
------, ``Discrete-value group and fully connected architectures for beyond
  diagonal reconfigurable intelligent surfaces,'' \emph{arXiv preprint
  arXiv:2110.00077v3}, 2021.

\bibitem{li2022reconfigurable}
Q.~Li, M.~El-Hajjar, I.~A. Hemadeh, A.~Shojaeifard, A.~Mourad, B.~Clerckx, and
  L.~Hanzo, ``Reconfigurable intelligent surfaces relying on non-diagonal phase
  shift matrices,'' \emph{IEEE Trans. Veh. Technol.}, vol.~71, no.~6, pp.
  6367--6383, 2022.

\bibitem{xu2021simultaneously}
J.~Xu, Y.~Liu, X.~Mu, J.~T. Zhou, L.~Song, H.~V. Poor, and L.~Hanzo,
  ``Simultaneously transmitting and reflecting {(STAR)} intelligent
  omni-surfaces, their modeling and implementation,'' \emph{IEEE Veh. Technol.
  Mag.}, vol.~17, no.~2, 2022.

\bibitem{zhang2020beyond}
S.~Zhang, H.~Zhang, B.~Di, Y.~Tan, Z.~Han, and L.~Song, ``Beyond intelligent
  reflecting surfaces: Reflective-transmissive metasurface aided communications
  for full-dimensional coverage extension,'' \emph{IEEE Trans. Veh. Technol.},
  vol.~69, no.~11, pp. 13\,905--13\,909, 2020.

\bibitem{zeng2022intelligent}
S.~Zeng, H.~Zhang, B.~Di, Y.~Liu, M.~Di~Renzo, Z.~Han, H.~V. Poor, and L.~Song,
  ``Intelligent omni-surfaces: Reflection-refraction circuit model,
  full-dimensional beamforming, and system implementation,'' \emph{IEEE Trans.
  Commun.}, vol.~70, no.~11, pp. 7711--7727, 2022.

\bibitem{li2022generalized}
H.~Li, S.~Shen, and B.~Clerckx, ``Beyond diagonal reconfigurable intelligent
  surfaces: From transmitting and reflecting modes to single-, group-, and
  fully-connected architectures,'' \emph{IEEE Trans. Wireless Commun.},
  vol.~22, no.~4, pp. 2311--2324, 2023.

\bibitem{li2022beyond}
------, ``Beyond diagonal reconfigurable intelligent surfaces: A multi-sector
  mode enabling highly directional full-space wireless coverage,'' \emph{IEEE
  J. Sel. Areas Commun.}, accepted, 2023.

\bibitem{shen2018fractional}
K.~Shen and W.~Yu, ``Fractional programming for communication systems--{Part
  I}: Power control and beamforming,'' \emph{IEEE Trans. Signal Process.},
  vol.~66, no.~10, pp. 2616--2630, 2018.

\bibitem{absil2009optimization}
P.-A. Absil, R.~Mahony, and R.~Sepulchre, ``Optimization algorithms on matrix
  manifolds,'' in \emph{Optimization Algorithms on Matrix Manifolds}.\hskip 1em
  plus 0.5em minus 0.4em\relax Princeton University Press, 2009.

\end{thebibliography}

\end{document}